\documentclass[12pt]{amsart}
\usepackage{amssymb,amsmath,natbib,tikz}
\usepackage{comment}
\usepackage{pgfplots}
\usepackage{tikz-3dplot}

\usetikzlibrary{intersections,positioning}
\usetikzlibrary{decorations.pathreplacing,angles,quotes,calc}

\usepackage[left=1.6in,top=1.8in,right=1.6in,bottom=1.8in]{geometry}
\usepackage[onehalfspacing]{setspace}
\usepackage[foot]{amsaddr}

\def\la{\lambda}
\def\al{\alpha}

\def\g{\gamma}
\def\ta{\theta}
\def\spn{\mbox{span}}

\def\then{\Longrightarrow}

\def\Na{\mathbf{N}}
\def\Re{\mathbf{R}}

\newcommand{\df}[1]{\textit{\textbf{#1}}}

\newcommand{\norm}[1]{ \| #1 \| }

\theoremstyle{plain}
\newtheorem{theorem}{Theorem}%[section]
\newtheorem{lemma}[theorem]{Lemma}%[section]
\newtheorem{proposition}[theorem]{Proposition}
\newtheorem{corollary}[theorem]{Corollary}
\newtheorem*{remark}{Remark}

\begin{document}
\title{Spherical preferences}

\author[Chambers]{Christopher P. Chambers}
\author[Echenique]{Federico Echenique}

\address[Chambers]{Department of Economics, Georgetown University}
\address[Echenique]{Division of the Humanities and Social Sciences,
  California Institute of Technology}
\email[A1,A2]{cc1950@georgetown.edu, fede@caltech.edu}

\thanks{Echenique thanks the National Science Foundation for its support
through the grants SES 1558757 and CNS 1518941.}

\begin{abstract}
We introduce and study the property of orthogonal independence, a restricted additivity axiom applying when alternatives are orthogonal. The axiom requires that the preference for one marginal change over another should be maintained after each marginal change has been shifted in a direction that is orthogonal to both. 

We show that continuous preferences satisfy orthogonal independence if and only if they are spherical: their indifference curves are spheres with the same center, with preference being ``monotone'' either away or towards the center. Spherical preferences include linear preferences  as a special (limiting) case. We discuss different applications to economic and political environments. Our result delivers Euclidean preferences in models of spatial voting, quadratic welfare aggregation in social choice, and expected utility in models of choice under uncertainty. %As an extension, we discuss an endogenous notion of orthogonality.
\end{abstract}

\maketitle

\section{Introduction}
We introduce and study the property of {\em orthogonal additivity,} or {\em orthogonal independence,} in choice theory, and find that it characterizes a class of preferences with spherical indifference curves. The property is simple to state. Imagine an agent choosing among consumption bundles: vectors in $\Re^n$. Such vectors can be interpreted in different ways to capture various economic environments.  Suppose an agent starts from an endowment, or status quo point, of $w$.  The agent is choosing to either shift her consumption from $w$ to $w+x$, or from $w$ to $w+y$. The axiom, which we term \textbf{Origin-independent orthogonal additivity (OIOI)}, says that \[ 
w+x\succeq w+y \text{ and } z\perp x,y \then w+(x+z)\succeq w+(y+z).
\] The notation $z\perp x,y$ means that direction $z$ is orthogonal to both $x$ and $y$. In a sense, it complements and substitutes $x$ and $y$ equally. The axiom says that the comparison of $x$ and $y$ should not be affected by the addition of the orthogonal direction~$z$.  The axiom is required to hold for every $w$, $x$, $y$, and $z$ satisfying the hypotheses.

Our main result is that OIOI has strong implications, though much weaker than the analogous unqualified version of independence would have. Together with continuity, OIOI implies \df{spherical preferences}: preferences with linear or spherical indifference curves.  If the preference has spherical indifference curves, each sphere must have the same center, and the preference must be monotone along any ray emanating from that center.  Examples of spherical preferences include perfect substitutes in consumption theory, expected utility in choice under uncertainty, and Euclidean preferences in voting theory.

We now outline several different economic environments where either OIOI has a natural meaning, or the spherical representation has particular interest.

\begin{itemize}
    \item Net trades. A consumer chooses among consumption bundles $x\in\Re^n$, which can be thought of as net trades as they involve negative quantities. Orthogonality has an intuitive geometric meaning.
    \item Spatial choice. A voter chooses among policy proposals. There are $n$ issues in question, and each policy proposal takes a stand on each issue, so that proposals can be represented as vectors in $\Re^n$. Spherical preferences are closely related to Euclidean preferences, which have received a lot of attention in the literature on voting \citep{downs1957economic,stokes1963spatial}.  In fact, Euclidean preferences are the special case of spherical preferences where there is an ``ideal point,'' and the individual is worse off the further away from the ideal point.
    \item Choice under uncertainty. An agent chooses among uncertain monetary payoffs (monetary acts). There are $n$ states of the world and each vector $x\in\Re^n$ represents a stage-contingent payoff. When $x$, $y$, and $z$ are non-negative, then $z\perp x$ and $z\perp y$ means that $z$ complements $x$ and $y$ in the same states. Thus $z$'s relation to the uncertainty inherent in $x$ is the same as its relation to the uncertainty inherent in $y$, and we may infer that $z$ is as good as hedge for $x$ as for $y$. 
    
    For choice under uncertainty it is natural to require a monotonicty axiom, in addition to continuity and OIOI. The objects of choice are monetary acts, so monotonicity is a natural property. Under these axioms we obtain (risk neutral) subjective expected utility.

    \item Social choice. Consider a society of $n$ agents, and interpret vectors in $\Re^n$ as reflecting the level of welfare of each individual agent. Linear preferences embody a form of utilitarianism \citep{harsanyi1955cardinal}. More general spherical preferences in this environment have been studied by several authors \citep{epstein1992quadratic}.
    \item Dispersion. Consider a finite set of states of the world, with a uniform probability measure over them.  The set of vectors which sum to zero are now mean-zero random variables---or monetary acts, and they form a well-defined finite-dimensional vector space.  Since all acts have mean zero, we can interpret a ranking as a measure or riskiness or dispersion.  Orthogonality now becomes the statement that two random variables are uncorrelated.  So the axiom then requires that the addition of an act which is uncorrelated with each of two additionally present acts will not reverse their ranking.  In this environment, our axiom becomes related to \citet{pomatto} and \citet{pomattotamuz2019}, except that we explore the stronger condition of zero correlation rather than statistical independence.  %We also work with a model including states of the world, whereas their work is on the space of distribution functions.
\end{itemize}

A key property here is that OIOI is a universal property:  it claims a relationship to hold for all collections satisfying certain hypotheses.  As such, and according to \cite{ces}, it is falsifiable.  On the other hand, the model described by the axiom is apparently existential, relying on the existence of a sphere's center, or a linear direction. These ideas are fleshed out in Section~\ref{sec:finitedate}.

An extension of our work  establishes how one might endogenize a notion of orthogonality.  Different notions of orthogonality may permit more general quadratic transformations.  For example, instead of $x\cdot x=0$, we could identify orthogonality with $x\cdot Ax = 0$ holds for some symmetric $A$.  In Section~\ref{sec:endogenous}, we do exactly this.  Observe that $A$ need not be positive semidefinite, so ``orthogonality'' could be of a hyperbolic form. Importantly, the notion of orthogonality is derived from a utility function, so that orthogonality is obtained as instances where a conditional additivity property, like OIOI, holds.

Finally, our results establish that the set of continuous preferences satisfying OIOI (with the topology of closed convergence) is homeomorphic to a sphere. See Section~\ref{sec:topologyOIOI}.

\subsection{Related literature}

Many authors have studied Euclidean preferences and quadratic utility. We give a very brief overview of the literature, but it is fair to say that our result is quite different from the existing work. \cite{bogomolnaia2007euclidean} consider a profile of preferences over a finite set of alternatives, and study numbers $n$ for which these can be embedded into $\Re^n$ so that preferences are Euclidean. \cite{eguia2011foundations,eguia2013spatial} also studies the embedding problem, and considers expected utility preferences where the von-Neumann Morgenstern function has the Euclidean form for the chosen embedding.  He also axiomatizes Euclidean, and other preferences, in an expected utility framework.  He obtains separability by positing an additivity axiom of \citet{fishburn1970utility} that is meaningful in the context of lotteries, but not in our context.  \cite{azrieli} considers Euclidean preferences when there is a valence dimension and considers families of voters indexed by their ideal point. \cite{knoblauch2010recognizing} and \cite{peters2017recognising} study the algorithmic problem of recognizing whether preferences are Euclidean. \cite{degan2009voters} looks at the empirical implications of Euclidean preferences for voter data. \cite{henry2013euclidean} follows up on the paper by Degan and Merlo by providing a formal statistical test, and an identification strategy for Euclidean preferences. 

General polynomial (expected) utility was studied by \cite{machina}, who connects an $m$-order polynomial to preferences that only care about the first $m$ moments of the relevant uncertain act. In the social choice context, quadratic utility was introduced in a generalization of Harsanyi's theory of utilitarian aggregation by \cite{epstein1992quadratic}, who consider a sort of betweenness axiom. %Finally, \cite{gershkov2018monotonic} also look at endogenous notions of orthogonality, but the point is to understand the relation with incentives, and their analysis is very different from ours. 

\section{Model and main result}\label{sec:modelandmain}

\subsection{Model and notation}

The objects of choice, or \df{alternatives} are vectors in $\Re^n$, where $n \geq 3$. The inner product between two vectors is denoted by $x\cdot y = \sum_{i=1}^n x_iy_i$. Two alternatives $x$ and $y$ are \df{orthogonal} (or \df{perpendicular}) if $x\cdot y = 0$. In this case we write $x\perp y$. The \df{norm} of a vector $x$ is defined as, and denoted by, $\norm{x} = \sqrt{x\cdot x}$.  Given two vectors, $x\in \Re^n$ and $y\in \Re^m$, the notation $(x,y)$ refers to that vector in $\Re^{n+m}$ whose first $n$ coordinates coincide with $x$, and whose last $m$ coincide with $y$.

Choice behavior is modeled through a binary relation $\succeq$ on $\Re^n$, which dictates choice among pairs of alternatives in $\Re^n$.

\subsection{Axioms and main result}

Suppose that $w\in\Re^n$ is given as a starting, or endowment, point, and consider two alternative marginal changes $x$ and $y$ from $w$. Ultimately, the choice is between $w+x$ and $w+y$. Suppose further that  $w+x$ is deemed at least as good as $w+y$; we ask what happens when an additional marginal change $z$, \emph{orthogonal to both $x$ and $y$} is additionally appended.  Our axiom requires that $w + x + z$ be at least as good as $w + y + z$. In other words, the ranking of the two marginal changes should not be affected when we shift those changes in an orthogonal direction.  Since it imposes additivity, our axiom is similar in spirit to the independence axiom of von Neumann-Morgenstern, but it restricts the set of marginal changes to be those qualified by orthogonality.  

We may envision applying our model to a political setting, where $\Re^n$ might represent what is usually called ``policy space.’’ The vectors $x\in\Re^n$ represent proposals, or positions, along $n$ ``issues.'' In such a setting, $w$ would represent a status quo policy, and preferences would then be over certain changes from that status quo. 

%Orthogonality 

\medskip

\textbf{Origin independent orthogonal independence (OIOI):}  For all $w,x,y,z\in \Re^n$, if $z \perp x$ and $z \perp y$, then $w+x \succeq w+y$ iff $w+x + z \succeq w+y+ z$.

\medskip

The other two axioms are standard.

\medskip

\textbf{Continuity:} For all $x\in \Re^n$, the sets $\{y\in\Re^n: y \succeq x\}$ and $\{y\in\Re^n: x \succeq y\}$ are closed.

\medskip

\textbf{Weak order:} $\succeq$ is complete and transitive.\footnote{Complete:  For every $x,y\in\Re^n$, $x\succeq y$ or $y\succeq x$.  Transitive:  For all $x,y,z\in\Re^n$, $x \succeq y$ and $y \succeq z$ implies $x \succeq z$.}

Our main theorem says that continuous weak orders satisfy OIOI if and only if they can be represented by one of three classes of utility functions.

\begin{theorem}\label{thm:main}Suppose that $n \geq 3$.  Then a preference $\succeq$ satisfies OIOI, continuity, and weak order if and only if one of the following is true:
\begin{enumerate}
\item\label{it:linear} There is $u\in \Re^n$ for which $x\succeq y$ iff $u \cdot x \geq u \cdot y$
\item\label{it:euclidean} There is $x^*\in \Re^n$ for which $x \succeq y $ iff $\|x- x^*\| \leq \|y - x^*\|$
\item\label{it:antieuclidean} There is $x^* \in \Re^n$ for which $ x\succeq y$ iff $\|x -x^*\| \geq \|y-x^*\|$.
\end{enumerate}
\end{theorem}

\begin{remark}We may replace OIOI by an axiom requiring that for any $x,y\in \Re^n$ and any $d\perp (x-y)$, $x\succeq y$ iff $x+d \succeq y+d$.  In fact, this is the axiom we utilize in our proof.  Recall the notion of separability discussed in \citet{debreu}.  Debreu's notion requires that the preference between any two vectors in $\Re^n$ that have a common projection on some subset of coordinates, does not depend on what that common projection is.  The alternative axiom we propose is easily seen to be a strengthening of Debreu's separability in this sense.  Namely, it requires separability to hold independently of the choice of orthogonal basis with which we represent vectors in $\Re^n$.\end{remark}

Thus, OIOI essentially requires that a preference take one of these three forms. The first representation, in \eqref{it:linear} of Theorem~\ref{thm:main}, is a standard linear preference.  In fact, we have as a simple consequence of the theorem that:

\begin{corollary}\label{cor:monotonepref}
Suppose that $n \geq 3$ and that a preference $\succeq$ satisfies OIOI, continuity, and weak order, and that there is $z\in\Re^n$ such that for all $x\in \Re^n$, $x+z\succeq x$. Then there is $u\in \Re^n$ for which $x\succeq y$ iff $u \cdot x \geq u \cdot y$.
\end{corollary} In particular, if $\succeq$ satisfies a standard monotonicity axiom ($x\succeq y$ whenever $x\geq y$), then OIOI and continuity implies the existence of a linear representation.  Actually, the condition in Corollary~\ref{cor:monotonepref} can be significantly weakened.  It is enough to postulate that there is no point that is either a strict local maximum, or a strict local minimum.  Other sufficient additional conditions for linearity can similarly be based off of the non-compactness of weak lower and upper contour sets.

The second representation, statement~\eqref{it:euclidean} of Theorem~\ref{thm:main}, implies that preferences are \df{Euclidean:} there is an ideal point $x^*$, whereby preference is maximized.  All other points (consumption bundles, or acts) are compared with respect to the distance to the ideal point.  The further away is a point, the worse it is.  As discussed in the introduction, Euclidean preferences are heavily used in spatial models in political science, but they have applications elsewhere as well.  The next result says that if we add a property of ``strict convexity," then our axioms pin down Euclidean preferences.
\begin{corollary}\label{cor:stconvexity}
Suppose that $n \geq 3$ and that a preference $\succeq$ satisfies OIOI, continuity, and weak order, and that $x\succeq y$ and $x\neq y$ implies that $(1/2)x+(1/2)y\succ y$. Then there is $x^*\in \Re^n$ for which $x \succeq y $ iff $\|x- x^*\| \leq \|y - x^*\|$
\end{corollary}  The axiom of strict convexity is well known, and used in many areas of economics.

The last possibility in statement~\eqref{it:antieuclidean} is a kind of dual to the Euclidean idea.  Instead of an ideal point, there is a worst point $x^*$.  The further away from the worst point, the better.  This is a model which might explain ``NIMBY'' style-preferences.\footnote{NIMBY stands for ``not in my backyard.''}

With a slight abuse of terminology, we term these preferences \df{spherical} because they have spherical indifference curves, where we understand the linear preferences in \eqref{it:linear} as spherical because a line is like a limit of spheres with larger and larger radii.  Corollary~\ref{cor:monotonepref} says that linear preferences are the only spherical preferences that satisfy a basic monotonicity axiom.  Corollary~\ref{cor:stconvexity} makes the obvious point that the only strictly convex spherical preference is Euclidean.

%\begin{comment}
\subsection{A cardinal approach: endogenous orthogonality}\label{sec:endogenous}

One drawback of the previous approach is that it can be hard to ascribe meaning to the notion of orthogonal vectors. One may instead want orthogonality to be an endogenous condition that triggers additivity. Here we turn to a cardinal version of our exercise, where we start from a utility, or social welfare function, as the primitive object. This primitive is harder to justify and reason about than the ordinal approach in our main theorem, but it has the advantage  that we do not need an exogenous notion of orthogonality. Instead, orthogonality will be endogenous.

To fix ideas, consider a social choice framework. Consider a society of $n$ individuals that chooses among vectors that represent individual agents' welfare. Let $U:\Re^n\rightarrow \Re$ be a social welfare function. 

Our main axiom asks us to think of outcomes $w+x$ obtained by starting from a \df{status quo} $w$, and modifies it in the direction of $x$. Then we can use $U$ to evaluate a change in the direction of $x$, or a change in the opposite direction, $-x$. The axiom requires that this evaluation has to be the same regardless of the status quo.

\medskip

\textbf{Status quo independence:} 
\[   \frac{1}{2}\left[U(w+x) - U(w) \right] + \frac{1}{2}\left[U(w - x) -
  U(w) \right] 
\] is independent of (or constant in) $w$. 

\medskip

Status quo independence says that a lottery that ``shorts''  and ``longs'' $x$ with
equal probability has cardinal gain that is independent of the
status quo $w$.

\medskip

\textbf{Eventual linearity:} 
For any $x$ and $y$ there is $w$ such that 
\[
U(w+(x+y)) - U(w-(x+y)) = U(w+x) - U(w-x) + U(w+y) - U(w-y)
\]

\medskip

Think of eventual linearity as suggesting a notion of orthogonality, and using the idea behind OIOI. The difference $\delta_x^w = U(w+x)-U(w-x)$ can be interpreted as a ``marginal utility'' in the direction of $x$ from the status quo $w$. Now, given $x$ and $y$ we can always find $w$ such that $(x-w)\perp (y-x)$, and OIOI as an ordinal axiom imposes a form of additivity for orthogonal directions. A cardinal version of OIOI might require that,  relative to the status quo $w$, the marginal change in utility $\delta_{x+y}^w$ should equal the sum of $\delta_{x}^w$ and $\delta_{y}^w$. Eventual linearity imposes this idea, without insisting on the standard Euclidean notion of orthogonality.

\begin{theorem}\label{thm:cardinal} Let $U$ be continuous and satisfy $U(0)=0$. Then $U$
  satisfies status quo independence  and eventual linearity iff $U = f+g$, where $f$ is quadratic and g
  is linear. Moreover, $f$ and $g$ are uniquely identified from~$U$.
\end{theorem}

That $f$ is quadratic means that there is a symmetric and bilinear function $S$ such that $f(x)=S(x,x)$. Observe that, as a consequence, we obtain an endogenous notion of orthogonality. We say that $x$ and $z$ are $U$-orthogonal whenever $S(x,z)=S(z,x)=0$. As a special case we have the conventional definition of orthogonality used in OIOI, with $S(x,z) = x\cdot z$. 

This means that if $x$ and $z$ and $U$-orthogonal, then $U(x+z) = S(x+z,x+z)+g(x+z) = S(x,x+z)+S(z,x+z)+g(x) + g(z)=U(x)+U(y)$. Thus $U$ satisfies a \textit{conditional} linearity property, in the same spirit as OIOI. Linearity is conditional on $U$-orthogonality.  For any $x$ and $y$, if $z$ is $U$-orthogonal to both $x$ and $y$, then $U(x)- U(y) = U(x+z)- U(y+z)$.

\section{Finite data and testing}\label{sec:finitedate}

Here, we imagine we have two binary relations $R$ and $P$, each of which are \emph{finite}, in the sense that $|P|,|R|<+\infty$.  We ask when there is a preference $\succeq$ of the form described in Theorem~\ref{thm:main} for which 
\begin{enumerate}
\item If $x \mathbin{R}y$, then $x \succeq y$
\item If $x \mathbin{P}y$, then $x \succ y$.
\end{enumerate}

In case there is such a $\succeq$, we say that $(R,P)$ are \emph{rationalizable by a spherical preference}.  In the following, $\Delta(R\cup P) \equiv \{\lambda\in\Re^{R\cup P}_+:\sum_{(x,y)\in R\cup P}\lambda(x,y)=1\}$.

The following is a counterpart of \citet{ce}, Theorem 11.11.  Rationalizability by spherical preferences turns out to be characterized by the satisfaction of a collection of linear inequalities.  The ``unknowns’’ in the linear inequalities are unobservable parameters for a particular utility representation of spherical preferences, and each observation from $R \cup P$ corresponds to one linear inequality.  Observations in $R$ are weak, while observations in $P$ are strict.  The result is a standard application of duality techniques.  In the statement of the Proposition, the vector $\lambda$ is intended to be a kind of dual variable, or Lagrange multiplier.  

The importance of the following proposition is the following.  Suppose we seek to falsify the model of spherical preferences.  A naive method of so doing would be to check every single possible spherical preference, and verify that this preference is inconsistent with the observed data.  Such a method is theoretically impossible, as there are an infinite number of spherical preferences.  The following proposition establishes that it is enough to demonstrate the existence of a collection of dual variables satisfying some conditions in order to falsify the model.  

\begin{proposition}$(R,P)$ are rationalizable by a spherical preference if and only if for any $\lambda\in\Delta(R\cup P)$ for which $\sum_{(x,y)\in P}\lambda(x,y)>0$, one of the following is true
\begin{enumerate}
\item $\sum_{(x,y)\in R\cup P}\lambda(x,y)(x\cdot x) \neq \sum_{(x,y)\in R\cup P}\lambda(x,y)(y\cdot y)$
\item $\sum_{(x,y)\in R\cup P}\lambda(x,y)x \neq \sum_{(x,y)\in R\cup P}\lambda(x,y)y$.
\end{enumerate}
\end{proposition}

\begin{proof}Rationalizability by a spherical preference is equivalent to the existence of $c\in \Re$ and $u\in \Re^n$ for which:

\[x\mathbin{R}y \rightarrow c(x\cdot x-y\cdot y) + u\cdot (x-y) \geq 0\]
\[x\mathbin{P}y \rightarrow c(x\cdot x-y\cdot y) + u \cdot (x-y) > 0.\]

This is a finite system of linear inequalities, whose consistency is equivalent to the condition in the statement of the Proposition.  See \citet{ce}, Lemma 1.12.
\end{proof}

\section{On the topological structure of the set of OIOI preferences}\label{sec:topologyOIOI}

Consider the set of preferences axiomatized in Theorem~\ref{thm:main}.  We claim that for $\Re^n$, upon removing the preference that is total indifference, the set of such preferences becomes homeomorphic to $S^n \equiv \{y\in \Re^{n+1}:\|y\|=1\}$.  To do so, we discuss a particular topology, the topology of closed convergence.  This topology is defined on the set of all binary relations $\succeq$ which satisfy the following three properties:
\begin{enumerate}
\item $\{(x,y):x\succeq y\}$ is closed.
\item $\succeq$ is \emph{reflexive}; \emph{i.e.} for all $x$, $x\succeq x$.
\item $\succeq$ is \emph{negatively transitive}; \emph{i.e.} $x\succ y\succ z$ implies $x\succ z$.
\end{enumerate}
Call this set of binary relations $\mathcal{P}$.  

 The topology of closed convergence, $\tau_c$, is the smallest Hausdorff topology on $\mathcal{P}$ for which the set $\{(x,y,\succeq):x \succeq y\}\subseteq \Re^n\times\Re^n \times \mathcal{P}$ is closed.\footnote{Recall a topology is \emph{Hausdorff} if every pair of distinct points can be separated by disjoint open sets.}  For example, see \citet{hildenbrand}, Theorem 1 p. 96.

Let the set $\Pi$ denote the set of all preferences axiomatized in Theorem~\ref{thm:main}, endowed with the topology of closed convergence.  The preference $\mathcal{I}$ represents complete indifference.

An important consequence of the following is that the set of OIOI preferences forms a compact set.  \citet{CEL2019} establishes that compactness of a set of preferences is a sufficient condition for ``recovering’’ a preference with finite data.

\begin{theorem}$\Pi\setminus\{\mathcal{I}\}$ is homeomorphic to $S^n$.\end{theorem}

\begin{proof}Observe that each $\succeq\in\Pi\setminus\{\mathcal{I}\}$ has a unique representation via:
\[u_{\succeq}(x) = c(x\cdot x) + d\cdot x,\]  where $c\in \Re$, $d\in\Re^n$, and $(c,d)\in S^n$, and that this map is one-to-one.

Further observe that $S^n$ is compact, and that the topology of closed convergence is Hausdorff, compact, metrizable (Corollary 3.81 of \citet{aliprantis1999}).

Finally, we show that the map $\pi:S^n \rightarrow \Pi$ whereby $\pi(x,d)$ is the preference represented by $c(x\cdot x)+d\cdot x$ is continuous, and then apply Theorem 2.33 of \citet{aliprantis1999}.  Continuity of the map $\pi$ follows easily from Theorem 8 of \citet{bordersegal}, using the fact that each $\succeq\in\Pi$ is locally strict.
\end{proof}

\section{Intuition behind Theorem~\ref{thm:main}}\label{sec:intuition}

We give a simple geometric intuition behind our main theorem. Our goal is to illustrate how the main force of the axiom implies linear indifference curves {\em on spheres.} Specifically, for each $w$ there is a vector $p_w$ such that for any sphere $S$ centered at $w$, if $x,y\in S$, then $x\sim y$ if and only if $p_w\cdot x = p_w\cdot y$. This is not quite enough to prove the theorem, but it serves to illustrate some of the forces behind it.\footnote{We should emphasize that the actual proof of Theorem~\ref{thm:main} relies on a completely different argument.}

One piece of notation we shall use is that,  for $x,y\in\Re^n$, $l(x,y) = \{\la x + (1-\la) y:\la\in \Re\}$ denotes the line passing through $x$ and $y$.

We shall use a seemingly stronger property than OIOI, namely:

\medskip

\textbf{Strong origin independent orthogonal independence (SOIOI):}  For all $x,y,a,b,w\in\Re^n$, if $x \perp y$, $a \perp b$, $(w + x) \succeq (w + a)$ and $(w + y) \succeq (w + b)$, then $(w + x + y) \succeq (w + a + b)$, with strict preference if either of the antecedent rankings are strict.

\medskip

One implication of Theorem~\ref{thm:main} is that SOIOI is not actually stronger than OIOI, at least under the remaining axioms. For the purpose of the arguments developed in this section, we use SOIOI because it implies a kind of homotheticity:

\begin{proposition}\label{prop:homoth}If $\succeq$ satisfies weak order, continuity, SOIOA, and $n \geq 3$, then for any $w,x,y\in\Re^n$ for which $\|x\|=\|y\|$, and any $\beta > 0$, $w + x \succeq w +y$ iff $w + \beta x \succeq w + \beta y$.\end{proposition}

The proof of Proposition~\ref{prop:homoth} is in Section~\ref{sec:proofprophomoth}.

\subsection{The case of $n=2$}

The first bit of intuition can be seen on the plane, that is with $n=2$.  The preference $\succeq$ is a continuous weak order, so it has a continuous utility representation $U$. Let $S$ be the sphere with center $0$ and radius $r>0$ on the plane. Write the sphere in polar coordinates, as \[ 
S=\{(\ta,r):0\leq \ta\leq 2\pi\}. 
\] We use addition mod $2\pi$ for angles.

First notice that there must exist two points $x=(\ta_x,r)$ and $y=(\ta_y,r)$ that are \df{antipodal} in the sense that $\ta_x=\ta_y+\pi$, and for which $U(x)=U(y)$. To see this suppose, if $U(0,r)=U(\pi,r)$, we are done so suppose (without loss of generality) that $U(0,r)>U(\pi,r)$ and consider the function \[g(t) = U(t,r) - U(t+\pi,r):[0,\pi]\rightarrow \Re.\] Then $g(0)>0>g(\pi)$. Hence, by the intermediate value theorem, there is $\ta\in[0,\pi]$ with $U(\ta,r)= U(\ta+\pi,r)$. 

Consider the following illustration, which is drawn with $w=0$ for simplicity:

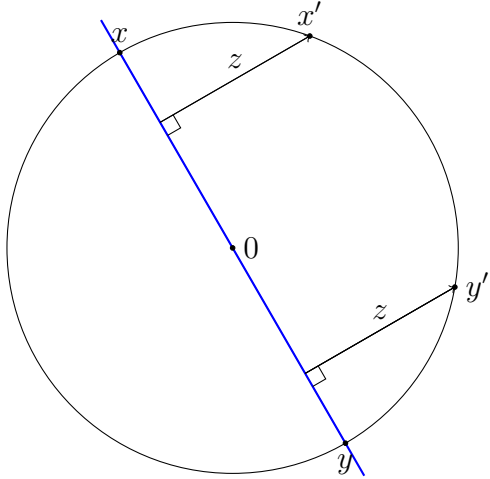
\begin{tikzpicture}[scale=1]
%circle and origin at w
\draw (0,0) circle (3) ;
\coordinate (w) (0,0);
\draw[blue,thick] (w) -- (120:3.5);
\draw[blue,thick] (w) -- (300:3.5);

% w+x and w+y
\draw[blue,name path=lup] (w) -- (120:3) coordinate (x) ;
\draw[blue,name path=ldown] (w) -- (300:3) coordinate (y) ;

% w+x' and w+y'
\path (0,0) -- (70:3) coordinate (xp);
\path (0,0) -- (350:3) coordinate (yp);

% projections
\path[name path=lxp] (xp) -- +(210:3);
\path [name intersections={of=lup and lxp,by=xpproject}];
\path[name path=lyp] (yp) -- +(210:3);
\path [name intersections={of=ldown and lyp,by=ypproject}];

\draw[thin] (xp) -- (xpproject)  ++(-60:.2) -- ++(30:.2) -- ++(120:.2);
\draw (ypproject) -- (yp);
\draw[thin] (yp) -- (ypproject)  ++(-60:.2) -- ++(30:.2) -- ++(120:.2);

%\draw[decoration={brace,mirror,raise=2pt},decorate] (yp) -- (ypproject) ;
\draw[->] (ypproject) -- node[above] {$z$} (yp) ;
\draw[->] (xpproject) -- node[above] {$z$} (xp) ;

%labels and dots
\filldraw (w) circle (.03); \node [right] (w) {$0$};
\filldraw (x) circle (.03); \path (x) node[above] {$x$};
\filldraw (y) circle (.03); \path (y) node[below] {$y$};
\filldraw (xp) circle (.03); \path (xp) node[above]  {$x'$};
\filldraw (yp) circle (.03); \path (yp) node[right] {$y'$};

\end{tikzpicture}

%%% Local Variables:
%%% mode: latex
%%% TeX-master: t
%%% End:
\label{circle1}

We show that indifference curves on $w+S$ are linear. By the previous argument, there exist $x$ and $y$, antipodal points on $S$, with the property that $w+x\sim w+y$. Consider the points $x',y'$ that lie on a line parallel to $l(x,y)$. Then there is $z\perp l(x,y)$ for which $x'-z$ and  $y'-z$ are on $l(x,y)$. 

Now, since $(x’-z)\perp  z$ and $(y’-z)\perp z$, given that $\norm{x’}=\norm{y’}$, we have $\norm{x'-z} = \norm{y'-z}$. So there is $\beta\in \Re$ with $x'-z = \beta x$ and $y'-z = \beta y$. Hence Proposition~\ref{prop:homoth} implies that  $x'-z\sim y'-z$. Then by OIOI, $x'\sim y'$.

\subsection{$n\geq 3$}

Consider a sphere $S$ with center $w$ and radius $r$. Choose $x_1$ and $x_2$, orthogonal vectors with $\norm{x_1}=\norm{x_2}=r$. Consider the {\em equator} defined by $x_1$ and $x_2$ on $S$: the set of points on the linear span of $\{x_1,x_2\}$ that have norm $r$. By the argument for $n=2$ there exists a pair of antipodal vectors $x$ and $y$ on the equator such that $w+x\sim w+y$.

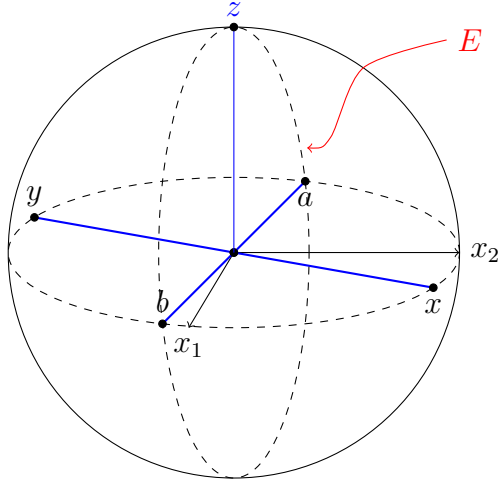
\begin{tikzpicture}
    % Define radius
    \def\r{3}

%origin
%\path (0,0) node[circle,fill,inner sep=1] (orig) {};
\path (0,0) coordinate (orig);

    % Sphere
    \draw (orig) circle (\r);
    \draw[dashed, name path=equator] (orig) ellipse (\r{} and \r/3);

    % Axes
    \draw[->] (orig) -- +(-\r/5,-\r/3) node[below] (x1) {$x_1$};
    \draw[->] (orig) -- +(\r,0) node[right] (x2) {$x_2$};

%antipodal points
\path[name path=tunel1] (orig) -- +(350:5);
\path[name path=tunel2] (orig) -- +(170:5);
\coordinate [name intersections={of=tunel1 and equator,by=punto1}];
\coordinate [name intersections={of=tunel2 and equator,by=punto2}];
\draw [blue,thick,text=black] (punto1) node[below] {$x$}-- (punto2) node[above] {$y$};

\draw [dashed] (orig) ellipse (\r/3 and \r{});

\filldraw (orig) circle (.05);\filldraw (punto1) circle (.05);\filldraw (punto2) circle (.05);

%antipodal points orthogonal to x-y
\path[name path=perptunel1] (orig) -- +(45:5);
\path[name path=perptunel2] (orig) -- +(225:5);
\coordinate [name intersections={of=perptunel1 and equator,by=asunto1}];
\coordinate [name intersections={of=perptunel2 and equator,by=asunto2}];
\draw [blue,thick,text=black] (asunto1) node[below] {$a$}-- (asunto2) node[above] {$b$};

\filldraw (orig) circle (.05);\filldraw (asunto1) circle (.05);\filldraw (asunto2) circle (.05);

\draw[blue] (orig) -- (90:\r) node[above] {$z$};
\filldraw (90:\r) circle (.05);

\draw[thin,red,<-] plot[smooth] coordinates {(55:1.7) (50:2) (55:3) (45:4)} node[right] {$E$};

%    %Angles
%    \pic [draw=gray,text=gray,->,"$\phi$"] {angle = x1--orig--phi};
%    \pic [draw=gray,text=gray,<-,"$\theta$"] {angle = a--orig--x3};

\end{tikzpicture}

%%% Local Variables:
%%% mode: latex
%%% TeX-master: t
%%% End:

Choose $a$ and $b$ on the equator such that $a$ and $b$ are antipodal, and perpendicular to $x$ and $y$. This is possible because the equator has dimension 2, and $x$ and $y$ are antipodal. Moreover, choose a vector $z$ that is orthogonal to the span of $\{x_1,x_2\}$. Consider the equator $E$ on $S$ defined by $a, b$ and $z$. 

On $E$ we must have, by the argument for $n=2$, two antipodal points $x'$ and $y'$ with $w+x'\sim w+y'$. Importantly,  $x'$ and $y'$ are perpendicular to  $x$ and $y$. Let $E'$ be the equator defined by $x$ and $x'$. $E'$ is two dimensional and generated by the orthogonal lines $l(x,y)$ and $l(x',y')$. 

The equator $E'$ is represented in the following figure. We shall prove that all the points on $E'$ are indifferent. 

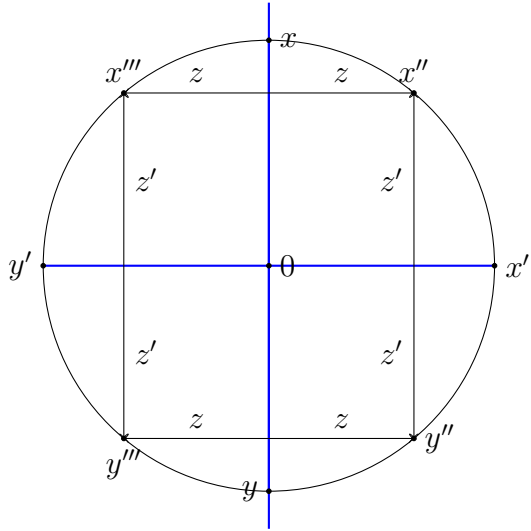
\begin{tikzpicture}[scale=1]
%circle and origin at w
\draw (0,0) circle (3) ;
\coordinate (w) (0,0);
\draw[blue,thick] (w) -- (90:3.5);
\draw[blue,thick] (w) -- (270:3.5);

% w+x and w+y
\draw[thick,blue,name path=lup] (w) -- (90:3) coordinate (x) ;
\draw[thick,blue,name path=ldown] (w) -- (270:3) coordinate (y) ;

% w+x' and w+y'
\draw[thick,blue,name path=morena] (w) -- (0:3) coordinate (xp) ;
\draw[thick,blue,name path=manya] (w) -- (180:3) coordinate (yp) ;

% w+x'' and w+y'' and w+x'''
\path (0,0) -- (50:3) coordinate (xpp);
\path (0,0) -- (130:3) coordinate (xppp);
\path (0,0) -- (310:3) coordinate (ypp);
\path (0,0) -- (230:3) coordinate (yppp);

% projections
\path[name path=lxpp] (xpp) -- +(180:3);
\path [name intersections={of=lup and lxpp,by=xppproject}];
\path[name path=xppmorena] (xpp) -- +(-90:3);
\path [name intersections={of=morena and xppmorena,by=xppprojectdown}];

\path[name path=lxppp] (xppp) -- +(270:3);
\path [name intersections={of=manya and lxppp,by=xpppproject}];
\path[name path=cavani] (xppp) -- +(0:3);
\path [name intersections={of=lup and cavani,by=xpppprojectright}];

\path[name path=lypp] (ypp) -- +(180:3);
\path [name intersections={of=ldown and lypp,by=yppproject}];
\path[name path=torrreira] (ypp) -- +(90:3);
\path [name intersections={of=morena and torrreira,by=yppprojectup}];

\path[name path=godin] (yppp) -- +(0:3);
\path [name intersections={of=ldown and godin,by=ypppprojectright}];
\path[name path=gimenez] (yppp) -- +(90:3);
\path [name intersections={of=manya and gimenez,by=ypppprojectup}];

%\draw[decoration={brace,mirror,raise=2pt},decorate] (yp) -- (ypproject) ;

\draw[->] (xppproject) -- node[above] {$z$} (xpp) ;
\draw[->] (xpppproject) -- node[right] {$z'$} (xppp) ;
\draw[->] (xppprojectdown) -- node[left] {$z'$} (xpp) ;

\draw[->] (xpppprojectright) -- node[above] {$z$} (xppp) ;

\draw[->] (yppproject) -- node[above] {$z$} (ypp) ;
\draw[->] (yppprojectup) -- node[left] {$z'$} (ypp) ;

\draw[->] (ypppprojectright) -- node[above] {$z$} (yppp) ;
\draw[->] (ypppprojectup) -- node[right] {$z'$} (yppp) ;

%labels and dots
\filldraw (w) circle (.03); \node [right] (w) {$0$};
\filldraw (x) circle (.03); \path (x) node[right] {$x$};
\filldraw (y) circle (.03); \path (y) node[left] {$y$};
\filldraw (xp) circle (.03); \path (xp) node[right]  {$x'$};
\filldraw (yp) circle (.03); \path (yp) node[left] {$y'$};
\filldraw (xpp) circle (.03); \path (xpp) node[above]  {$x''$};
\filldraw (xppp) circle (.03); \path (xppp) node[above]  {$x'''$};
\filldraw (ypp) circle (.03); \path (ypp) node[right] {$y''$};
\filldraw (yppp) circle (.03); \path (yppp) node[below]  {$y'''$};
\end{tikzpicture}

So consider first $x''$ and $y''$ that lie on a line parallel to $l(x,y)$. By the argument for $n=2$, $w+x''\sim w+y''$. Reflect $x''$ and $y''$ across the $l(x,y)$ line and consider the points $x'''$ and
$y'''$ on $E$. Again we obtain that $w+x'''\sim w+y'''$.

Note now that $y''$ and $y'''$ are the reflection of (respectively)
$x''$ and $x'''$ across the $l(x',y')$ line. Then $w+x'\sim w+y'$
means that $w+x''\sim w+x'''$ and $w+y''\sim w+y'''$. Hence we obtain
that \[
w+ x''\sim w+y'' \sim w+y''' \sim w+x''' .\]

This implies that any point on $E'$ is indifferent to its antipodal point. To see this, consider $a$ on the following figure and let $a'$ be its antipodal point.

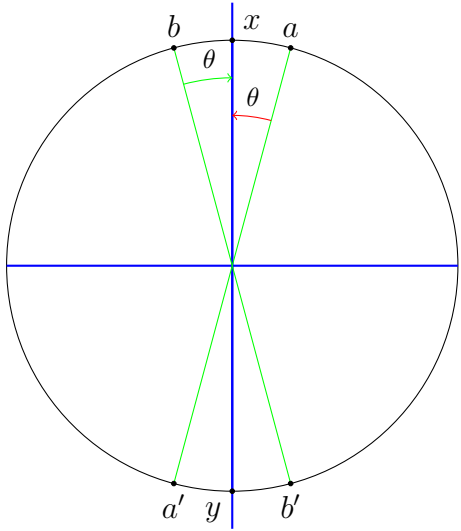
\begin{tikzpicture}[scale=1]
%circle and origin at w
\draw (0,0) circle (3) ;
\coordinate (w) (0,0);
\draw[blue,thick] (w) -- (90:3.5);
\draw[blue,thick] (w) -- (270:3.5);

% w+x and w+y
\draw[thick,blue,name path=lup] (w) -- (90:3) coordinate (x) ;
\draw[thick,blue,name path=ldown] (w) -- (270:3) coordinate (y) ;

% w+x' and w+y'
\draw[thick,blue,name path=morena] (w) -- (0:3) coordinate (xp) ;
\draw[thick,blue,name path=manya] (w) -- (180:3) coordinate (yp) ;

% w+x'' and w+y'' and w+x'''
\path (0,0) -- (75:3) coordinate (a); \path (0,0) -- (255:3) coordinate (antia);
\path (0,0) -- (105:3) coordinate (b); \path (0,0) -- (285:3) coordinate (antib);

\draw[green] (a) -- (antia) ;
\draw[green] (antib) -- (b);

\filldraw (a) circle (.03) node[above] {$a$}; \filldraw (antia) circle (.03) node[below] {$a'$};
\filldraw (b) circle (.03) node[above]  {$b$}; \filldraw (antib) circle (.03) node[below]  {$b'$};

\filldraw (x) circle (.03) node[above right] {$x$};
\filldraw (y) circle (.03) node[below left] {$y$};

\draw[thin,red,->] (75:2) arc (75:90:2); 
\path (82:2) node[above] {\footnotesize $\theta$} ;
\draw[thin,green,->] (105:2.5) arc (105:90:2.5);
\path (97:2.5) node[above] {\footnotesize $\theta$} ;

\end{tikzpicture}

Let $b$ be the reflection of $a$ across $l(x,y)$. By the previous
argument $w+a\sim w+b \sim w+a'$. So any point is indifferent to its
antipodal point.

Finally consider any two points on the same orthant of $E'$: Say $a$ and
$b$. Let $c$ be the vector $\frac{1}{2}(a+b)$, scaled to have norm $r$. Let $c'$ be the antipodal point to $c$
on $E$,  $d$ be perpendicular to $c$, and $d'$ be antipodal to $d$.

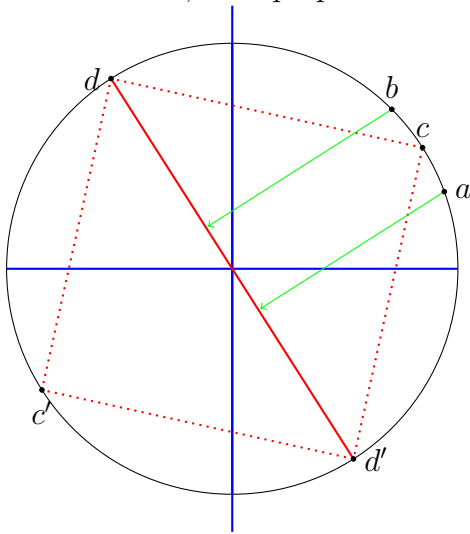
\begin{tikzpicture}[scale=1]
%circle and origin at w
\draw (0,0) circle (3) ;
\coordinate (w) (0,0);
\draw[blue,thick] (w) -- (90:3.5);
\draw[blue,thick] (w) -- (270:3.5);

% w+x and w+y
\draw[thick,blue,name path=lup] (w) -- (90:3) coordinate (x) ;
\draw[thick,blue,name path=ldown] (w) -- (270:3) coordinate (y) ;

% w+x' and w+y'
\draw[thick,blue,name path=morena] (w) -- (0:3) coordinate (xp) ;
\draw[thick,blue,name path=manya] (w) -- (180:3) coordinate (yp) ;

% w+x'' and w+y'' and w+x'''
\path (0,0) -- (20:3) coordinate (a);
\path (0,0) -- (45:3) coordinate (b);

%\path (0,0) -- (200:3) coordinate (antia);
%\path (0,0) -- (225:3) coordinate (antib);

\path (0,0) -- (32.5:3) coordinate (c);
\path (0,0) -- (122.5:3) coordinate (d);

\path (0,0) -- (212.5:3) coordinate (antic);
\path (0,0) -- (302.5:3) coordinate (antid);

%lines

\draw[red,thick,name path=perpcebolla] (antid) -- (d);

\draw[red,dotted,thick] (c) -- (d) -- (antic) -- (antid) -- cycle ;
%\draw[green,name path=punteado] (b) -- (antib);

\draw[green,->] (a) -- +(212.5:2.9);
\draw[green,->] (b) -- +(212.5:2.9);

%labels and dots

\filldraw (a) circle (.03) node[right] {$a$};
\filldraw (b) circle (.03) node[above]  {$b$};
\filldraw (c) circle (.03) node[above] {$c$};
\filldraw (d) circle (.03) node[left] {$d$};

\filldraw (antid) circle (.03) node[right] {$d'$};
\filldraw (antic) circle (.03) node[below] {$c'$};

\end{tikzpicture}

Then $w+d\sim w+d'$ as we have shown that antipodal points are
indifferent. This implies that $w+a\sim w+b$ by the same projection
argument as before.

Since $a$ and $b$ were arbitrary on the same orthant, we have that $w+a\sim w + b$ for all $a,b\in E'$.

The previous arguments establish the following.  In the proposition, $p$ is a vector which is orthogonal to $E’$, as derived in the preceding discussion.  Using the argument that indifference curves are linear on the circle, we get indifference for any pair of vectors on the sphere for which $p$ is normal to their difference, simply by taking the geodesic circle passing through the pair of vectors.

\begin{proposition} For each $w$ and $r$ there is $p\in\Re^n$ such that for $x,y\in S(w,r)$, $x\sim y$ iff $p\cdot x= p\cdot y$. 
\end{proposition}

Now it is easy to show 

 \begin{proposition} For each $w$ and $r$ there is $p\in\Re^n$ such that, for any $r'\leq r$ and $x,y\in S(w,r')$, $x\sim y$ iff $p\cdot x= p\cdot y$.
 \end{proposition}

\begin{proof} Let $p$ be as in the previous claim and 
  $r'\leq r$ and $\beta = r'/r$. Then $x,y\in S(w,r)$ iff $\beta
  x,\beta y\in S(w,r')$. Then
  $p\cdot x= p\cdot y$ iff 
  $w+x\sim w+y$ iff $w+\beta x\sim w+\beta y$ (an implication of Proposition~\ref{prop:homoth}).
\end{proof}

\section{Proof of Theorem~\ref{thm:main}}\label{sec:proofmain}
\subsection{Necessity}
We demonstrate that the three types of preferences satisfy OIOI. It is obvious that they are continuous weak orders.

So observe that any preference in the class has a representation as $u(x) = c x\cdot x + v\cdot x$, for some $c\in\Re$ and $v\in\Re^n$. Then  $(w + x) \succeq (w+y)$ implies that \[c (w + x)\cdot (w+x) + v\cdot (w+x) \geq c (w+y)\cdot (w+y)+ v \cdot (w + y).\] Add $c(w+z)\cdot (w+z)+v\cdot (w+z)$ to both sides to obtain that
\begin{align*}c (2w \cdot w + 2 w \cdot (x + z) + x\cdot x + z \cdot z) + v\cdot( 2w + x + z) \\
\geq c(2w\cdot w + 2w\cdot (y+z) + y\cdot y + z \cdot z) + v\cdot  (2w + y+z).\end{align*}
Subtract, from each side, $c w \cdot w + v \cdot w$, obtaining:
\begin{align*}c (w\cdot w + 2w\cdot (x+z) + x\cdot x + z \cdot z) + v \cdot (w + x + z) \\
\geq c( w\cdot w + 2w\cdot (y+z) +y \cdot y + z\cdot z)+ v\cdot (w+y+z).\end{align*}
Simplify and obtain:
\[c(w+x+z)\cdot (w+x+z) + v\cdot (w + x + z) \geq c (w+y+z) +v \cdot (w + y +z),\] using the fact that $x \perp z$ and $y \perp z$ (hence $x\cdot x+z\cdot z = (x+z)\cdot(x+z)$ and $y\cdot y+z\cdot z = (y+z)\cdot(y+z)$).  Therefore, $w + x+z \succeq w+ y+z.$

\subsection{Sufficiency}
\begin{proposition}\label{prop:newperpd}For a weak order satisfying OIOI, if $d \perp (x-y)$, then $x \succeq y$ iff $x + d \succeq y + d$.\end{proposition}

\begin{proof}Observe that $x + (0) \succeq x + (y-x)$.  Further, $d \perp (0)$ and $d \perp (y - x)$.  So by OIOI, $x + (0) + d \succeq x + (y-x) +d$, or $x + d\succeq y + d$.  Conversely, $-d \perp [(x+d)-(y+d)]$, so the result follows from the first step.  \end{proof}

Say a vector subspace $D$ of $\Re^n$ is \emph{inessential} if for any $d \in D$ and any $x \in \Re^n$, $x + d \sim x$.

\begin{proposition}\label{prop:newsubsp}If $\succeq$ satisfies OIOI, weak order, and continuity, then if it has a nontrivial inessential subspace, it is a linear preference.\end{proposition}

\begin{proof}To ease notation, suppose that the nontrivial inessential subspace is the subspace spanned by $(0,\ldots,0,1)$.

By definition of inessential, there is a preference $\succeq^*$ on $\Re^{n-1}$ for which $(x,c) \succeq (y,d)$ iff $x \succeq^* y$.  We claim that $\succeq^*$ is a linear preference, that is, for any $x,y,z\in \Re^{n-1}$, we have $x \succeq^* y$ iff $x + z \succeq^* y + z$.

Thus, suppose that $x \succeq^* y$ and let $z\in\Re^{n-1}$ be arbitrary.  Then for any $a\in \Re$, $(x ,0) \succeq (y,a)$.  In particular, let $a = z \cdot (x-y)$.  Observe that $(z,1) \perp (y-x,a) = (y,a)-(x,0)$.  Consequently by Proposition~\ref{prop:newperpd} we have $(x+z,1) \succeq (y+z,a+1)$, establishing that $x + z \succeq^* y + z$.

The remainder is now standard.  \end{proof}

\begin{corollary}\label{cor:additive}Suppose that $\succeq$ satisfies OIOI and weak order.  Suppose $n \geq 3$ and let $\{f_1,\ldots,f_n\}$ be an orthonormal basis for $\Re^n$.  For any $a,b\in\Re^n$ and any subset $G\subseteq \{1,\ldots,n\}$, we have: \[\sum_{i\in G}a_i f_i + \sum_{i\notin G}a_i f_i \succeq \sum_{i\in G}b_i f_i + \sum_{i\notin G} a_i f_i\] iff \[\sum_{i\in G}a_i f_i + \sum_{i\notin G}b_i f_i \succeq \sum_{i\in G}b_i f_i + \sum_{i\notin G} b_i f_i.\] \end{corollary}

\begin{proof}Follows from Proposition~\ref{prop:newperpd} by taking $x = \sum_{i\in G}a_i f_i + \sum_{i\notin G}a_i f_i $, $y = \sum_{i\in G}b_i f_i + \sum_{i\notin G} a_i f_i$, and $d = \sum_{i\notin G}(b_i-a_i)f_i$.\end{proof}

For our final step we need some additional notational conventions. For any subspace $T$ of $\Re^n$ and any  $x\in \Re^n$, let $\alpha_T(x)$ be the orthogonal projection of $x$ onto $T$.  If $T = \mbox{span}\{f\}$ for some vector $f\in \Re^n$, we abuse notation and write $\alpha_f(x)$ as the norm of $\alpha_{\mbox{span}\{f\}}$.

\begin{remark}
The final steps establish the following.  Let $S^{n-1}$ denote the unit sphere.  There is a utility representation $u$ of $\succeq$, and 
for each $f\in S^{n-1}$ a function $u_{f}:\Re\rightarrow\Re$,  satisfying the following properties:

\begin{enumerate}
\item For any orthonormal basis $\{f_1,\ldots,f_n\}$, if $x = \sum_{i=1}^n \alpha_{f_i}f_i$, then $u(x) = \sum_{i=1}^n  u_{f_i}(\alpha_{f_i})$.
\item $u(0) = 0$.
\end{enumerate}
It is then easy to show that $u:\Re^n \rightarrow \Re$ satisfies the property that if $ x \perp y$, then $u(x+y) = u(x) + u(y)$.  
\end{remark}

\begin{remark}The proof proceeds in Lemma~\ref{lem:rotation} and Proposition~\ref{prop:orthoadd} by establishing the result for $n=3$.  Then Proposition~\ref{prop:finalstep} extends the result to all $n \geq 3$.\end{remark}

\begin{lemma}\label{lem:rotation} Suppose that $n=3$, and that $\succeq$ is a continuous weak order satisfying OIOI.  Suppose that $\succeq$ has no non-trivial inessential subspaces.
Let $\{f_1,f_2,f_3\}$ be an orthonormal basis for $\Re^3$ and suppose that $u(z)=\sum_{i=1}^3 u_{f_i}(\al_{f_i}(z))$ is an additive representation for $\succeq$ for which $u_{f_i}(0)=0$ for each $i=1,2,3$. 
If $\{e_2,e_3\}$ is any other orthonormal basis for $\spn(\{f_2,f_3\})$, then there is an additive representation for $\succeq$,
\[v(z)=v_{f_1}(\al_{f_1}(z))+v_{e_2}(\al_{e_2}(z))+v_{e_3}(\al_{e_3}(z)),\] such that \[v(z)=u(z)=u_{f_1}(\al_{f_1}(z))+ v_{e_2}(\al_{e_2}(z))+v_{e_3}(\al_{e_3}(z)),\] and $v_{e_2}(0) = v_{e_3}(0)=0$. 
\end{lemma}
\begin{proof}
First observe that by Corollary~\ref{cor:additive} and \cite{debreu} (Theorem 3), since $\{f_1,e_2,e_3\}$ is an orthonormal basis, and there are no non-trivial inessential subspaces, $\succeq$ has an additive representation $v(z)=v_{f_1}(\alpha_{f_1}(z))+v_{e_2}(\alpha_{e_2}(z)) + v_{e_3}(\alpha_{e_3}(z))$.  We shall prove that we can choose this representation so that $v_{f_1}=u_{f_1}$,  $u=v$, $v_{e_2}(0)=0$, and $v_{e_3}(0)=0$.  

Define $T \equiv \spn (\{e_2,e_3\})=\spn(\{f_2,f_3\})$. We have two additive representations $u,v$ of $\succeq$:
\begin{align*}
u(z)=u(\alpha_{f_1}(z) + \al_T(z)) & = u_{f_1}(\alpha_{f_1}(z)) \\
& + 
[u_{f_2}(\alpha_{f_2}(\alpha_T(z))+u_{f_3}(\alpha_{f_3}(\alpha_T(z)))] \\  
v(z)=v(\alpha_{f_1}(z) + \al_T(z)) & = v_{f_1}(\alpha_{f_1}(z)) \\
& + [v_{e_2}(\alpha_{e_2}(\alpha_T(z))+v_{e_3}(\alpha_{e_3}(\alpha_T(z)))].
\end{align*}

The pair $(\spn(\{f_1\})\times T,\succeq)$ constitutes an additive conjoint measurement structure in the sense of \cite{FOM} (see Chapter 6.2.4).\footnote{See also \citet{fishburn1970utility}, Theorem~5.2 or~5.4} By their Theorem~2, there exists $\beta>0$, and $\g,\g'$ with $u_{f_1} = \beta v_{f_1} + \g$ and, and for every $x\in T$,  $u_{f_2}(\alpha_{f_2}(x))+u_{f_3}(\alpha_{f_3}(x)) = \beta(v_{e_2}(\alpha_{e_2}(x))+v_{e_3}(\alpha_{e_3}(x))) + \g'$. So define $v'_{f_1} = \beta v_{f_1} + \g$,  $v'_{e_2} = \beta v_{e_2} + \ta_2$ and $v'_{e_3} = \beta v_{e_3} +\ta_3$, where $\ta_2 =  -\beta v_{e_2}(0)$ and $\ta_3=-\beta v_{e_2}(0)$. Note that $0=(u_{f_2}+u_{f_3})(0) = \beta (v_{e_2}+v_{e_3})(0) + \g'$ implies that $\g'=\ta_2+\ta_3$.  Hence, we obtain that 
\begin{align*}
v'(z) & =  v'_{f_1}(\alpha_{f_1}(z))+v'_{e_2}(\alpha_{e_2}(z))+v'_{e_3}(\alpha_{e_3}(z)) \\
& = u_{f_1}(\alpha_{f_1}(z))+v'_{e_2}(\alpha_{e_2}(z))+v'_{e_3}(\alpha_{e_3}(z)) \\
& = u_{f_1}(\alpha_{f_1}(z))+u_{f_2}(\alpha_{f_2}(z))+u_{f_3}(\alpha_{f_3}(z)) 
= u(z).
\end{align*} while $v'_{f_1}(0) = v'_{e_2}(0) = v'_{e_3}(0)=0$. Thus $v'$ has the desired properties.
\end{proof}

\begin{proposition}\label{prop:orthoadd}Suppose that $n = 3$, that $\succeq$ is a continuous weak order satisfying OIOI, and has no non-trivial inessential subspaces.  Then there is a continuous utility representation $u:\Re^n\rightarrow \Re$ for which for any $x,y\in\Re^n$ with $x \perp y$, we have $u(x+y) = u(x) + u(y)$.\end{proposition}

\begin{proof}
We say that $x$ and $y$ are parallel, or collinear, if there is a scalar $\la\in \Re$ with $y=\la x$.

Let $\{f_1,f_2,f_3\}$ be a given orthonormal basis of $\Re^3$. By Corollary~\ref{cor:additive} and \citet{debreu} (Theorem 3), since there are no non-trivial inessential subspaces,  there exists a representation $u:\Re^n\rightarrow \Re$ of $\succeq$ for which $u(z) = \sum_{i=1}^3 u_{f_i}(\alpha_{f_i}(z))$.  Suppose without loss of generality that $u_{f_i}(0) = 0$, as additive representations are preserved by an additive translation of each component utility. 

Now, fix arbitrary $x,y\in \Re^n$ for which $x \perp y$.  If either $x$ or $y$ is $0$, then we know that $u(x+y)=u(x)+u(y)$ because $u(0)=0$. So lets suppose that $x,y\neq 0$.  Now, we have three possible cases to consider:
\begin{enumerate}
\item $x$ is parallel to some $f_i$ and $y$ is parallel to some $f_j$, 
\item Either $x$ or $y$ is parallel to some $f_i$, and the other one is not parallel to some $f_j$
\item Neither $x$ nor $y$ are parallel to any $f_i$.
\end{enumerate}

We shall prove that case 3 reduces to case 2, and that case 2 reduces to case 1. 

Let us first consider case $3$.  Note $\spn\{f_2,f_3\}\cap\spn \{x,f_1\}$ is nonempty, as $x$ is not collinear with $f_1$. So choose $e_2\in \spn\{f_2,f_3\}\cap\spn \{x,f_1\}$, scaled so that $\norm{e_2}=1$. Let $e_3\in \mbox{span}\{f_2,f_3\}$ then be a unit vector with $e_3\perp e_2$. Thus $x \in \mbox{span}\{f_1,e_2\}$ and  $\{f_1,e_2,e_3\}$ is an orthonormal basis of $\Re^3$.

By Lemma~\ref{lem:rotation}, there exists $v_{e_2}$ and $v_{e_3}$ such that 
\[u(z) = u_{f_1}(\alpha_{f_1}(z))+v_{e_2}(\alpha_{e_2}(z))+v_{e_3}(\alpha_{e_3}(z))\]
while $v_{e_2}(0)=0$ and $v_{e_3}(0)=0$.

Since $x\in \spn(\{f_1,e_2\})$ and $e_3\perp \spn(\{f_1,e_2\})$, there exists $e_1$ such that $\{e_1,\frac{1}{\norm{x}}x\}$ is an orthonormal basis for $\spn(\{f_1,e_2\})$ and $\{e_1,\frac{1}{\norm{x}}x,e_3\}$ is an orthonormal basis for $\Re^n$.

By Lemma~\ref{lem:rotation} again, there exists $u_{e_1}$ and $u_{\frac{x}{\norm{x}}}$ such that 
\[u(z) = u_{e_1}(\alpha_{e_1}(z)) + u_{\frac{x}{\|x\|}}(\alpha_{\frac{x}{\|x\|}}(z)) + v_{e_3}(\alpha_{e_3}(z)),\]
$u_{e_1}(0)=0$ and $u_{\frac{x}{\norm{x}}}(0)=0$. We are now in the situation of Case 2, as $x$ is parallel to $\frac{x}{\norm{x}}$.

So consider $x$ and $y$ in the configuration of Case 2. In particular suppose that $x$ is parallel to $f_2$. Since $x\perp y$, $y\in \spn(\{f_1,f_3\})$. So there exists a unit vector $w$ such that $\spn(\{y,w\}) = \spn(\{f_1,f_3\})$. By Lemma~\ref{lem:rotation} there exists $v_w$ and $v_{\frac{y}{\norm{y}}}$ such that 
\[u(z) = u_{\frac{x}{\|x\|}}(\alpha_{\frac{x}{\|x\|}}(z) )+ v_{\frac{y}{\|y\|}}(\alpha_{\frac{y}{\|y\|}}(z)) + v_w(\alpha_w (z)),\] with
$u_{\frac{x}{\|x\|}}(0) = v_{\frac{y}{\|y\|}}(0) = v_w(0) = 0$.  Thus, we are now in the situation of Case 1.

So consider $x$ and $y$ in the configuration of Case 1. In particular, suppose that $x$ is parallel to $f_1$ while $y$ is parallel to $f_2$. Recall that $x$ and $y$ are each nonzero. Observe that 
\begin{align*}
u(x+y) & = u_{\frac{x}{\|x\|}}(\|x\|) +u_{\frac{y}{\|y\|}}(\|y\|)+u_{f_3}(0) \\
& =  \left(u_{\frac{x}{\|x\|}}(\|x\|)+u_{\frac{y}{\|y\|}}(0) + u_{f_3}(0)\right) +  
\left(u_{\frac{x}{\|x\|}}(0)+u_{\frac{y}{\|y\|}}(\|y\|) + u_{f_3}(0)\right) \\
& = u(x) +u(y),
\end{align*}
where the penultimate equality follows from the fact that $0=u(0)=u_{\frac{x}{\|x\|}}(0)+u_{\frac{y}{\|y\|}}(0) + u_z(0)$.
\end{proof}

\begin{proposition}\label{prop:finalstep}Suppose that $n \geq 3$, and 
that $\succeq$ is a continuous weak order satisfying OIOI, and has no non-trivial inessential subspaces.  Then there exists $v\in\Re^n$ and a scalar $c$ such that $u(x) = c\norm{x}^2+v\cdot x$ is a utility representation of $\succeq$. \end{proposition}

\begin{proof}
For $n=3$ we have shown that there exists a utility representation that satisfies $u(x+y) = u(x) + u(y)$ for any  $x,y\in\Re^n$ with $x \perp y$. Then, by Theorem 1 of \citet{sundaresan},  $u(x) = c x\cdot x + v \cdot x$ for some $c\in\Re$ and $v\in\Re^n$.

So consider $n \geq 3$. By Corollary~\ref{cor:additive} and Theorem 3 in \citet{debreu}, and since there are no inessential non-trivial subspaces, there exists a utility representation 
\[U(x) = \sum_i w_i(x_i)\] of $\succeq$. 

By the preceding argument, for any subset $\{i,j,k\}\subseteq\{1,\ldots,n\}$ of cardinality 3, the representation restricted to $\Re^{i,j,k}$ can be chosen to be of the form \[u_{\{i,j,k\}}(x_{\{i,j,k\}})=c^{\{i,j,k\}}(x_i^2+x_j^2 + x_k^2) +  v^{\{i,j,k\}} \cdot (x_i, x_j, x_k).\] 

Then for any $\{i,j,k\}\subseteq\{1,\ldots,n\}$ of cardinality 3 we have two additive representations on $\Re^{\{i,j,k\}}$ : $w_i(x_i)+w_j(x_j)+w_k(x_j)$ and 
$\sum_{h\in \{i,j,k\}} c^{\{i,j,k\}}(x_h^2) + v^{\{i,j,k\}}_{h} x_h$. 
By Theorem 2 in  Chapter 6.2.4 of \cite{FOM}, there exists  $\al^{\{i,j,k\}}$ and $\beta^{\{i,j,k\}}>0$ with  \[\beta^{\{i,j,k\}} w_h(x_h) + \al^{\{i,j,k\}} = c^{\{i,j,k\}}(x_h^2) + v^{\{i,j,k\}}_{h} x_h\] for all $x_h$. This is true for all $x_h$ iff there is $\beta,$ $c_h$ and $v_h$ with 
$\beta^{\{i,j,k\}} = \beta>0$, $c^{\{i,j,k\}} = c$, $v^{\{i,j,k\}}_{h} = v_h$
and $\al^{\{i,j,k\}} = 0$.\footnote{Normalize $\beta^{\{i,j,k\}}=1$. Then, for $k\neq l$ we have 
$w_i(x_i) + \al^{\{i,j,k\}} = c^{\{i,j,k\}}(x_i^2) + v^{\{i,j,k\}}_{i} x_i$ and 
$w_i(x_i) + \al^{\{i,j,l\}} = c^{\{i,j,l\}}(x_i^2) + v^{\{i,j,l\}}_{i} x_i$. Hence,
$\al^{\{i,j,k\}} - \al^{\{i,j,l\}}  = (c^{\{i,j,k\}} - c^{\{i,j,l\}}) + (v^{\{i,j,k\}}_{i} - v^{\{i,j,l\}}_{i}) x_i$. This can only hold for all $x_i\in\Re$ if 
$\al^{\{i,j,k\}} - \al^{\{i,j,l\}} = c^{\{i,j,k\}} - c^{\{i,j,l\}} = v^{\{i,j,k\}}_{i} - v^{\{i,j,l\}}_{i} = 0$.}
Hence, $\beta w_h(x_h) = c x^2_h + v_h x_h$.

\end{proof}

\section{Proof of Theorem~\ref{thm:cardinal}}\label{sec:proofcardinal}

We prove sufficiency. So let $U$ be as in the statement of the theorem.
Let \[
f(x) = \frac{1}{2}\left[U(z+x) - U(z) \right] + \frac{1}{2}\left[U(z - x) - U(z) \right]
 \] and define $g(x) = U(x) - f(x)$. The following lemmas show sufficiency.

 \begin{lemma}\label{lem:one}
\[
f(x+y) + f(x-y) = 2f(x) + 2f(y) \]
    \end{lemma}
    
\begin{proof}
  \begin{align*}
A=    f(x-y) - f(x) - f(y)  
=& \frac{1}{2} U(q+ (x-y)) + \frac{1}{2} U(q- (x-y)) - U(q)  \\
&-\frac{1}{2} U(q'+ x) - \frac{1}{2} U(q' - x) + U(q')  \\
&-\frac{1}{2} U(q''+ y) - \frac{1}{2} U(q'' - y) + U(q'')  \\
=& \frac{1}{2} U(z+ (x-y)) + \frac{1}{2} U(z- (x-y)) - U(z)  \\
&-\frac{1}{2} U(z+ (x+y)) - \frac{1}{2} U(z - (x-y)) + U(z+y)  \\
&-\frac{1}{2} U(z-(x-y)) - \frac{1}{2} U(z-(x + y)) + U(z-x)  \\
  \end{align*}

Where the first equality is by definition of $f$, with arbitrary $q,q',q''\in\Re^n$. The second uses $q=z$, $q'=z+y$ and $q''=z-x$. 

Then we have that 
\[\begin{split}
A  =  -f(x+y) - 2 U(z) + U(z+y)  + U(z-x)  \\
+ \frac{1}{2} \left[U(z+ (x-y)) +  U(z- (x-y)) -  U(z - (x-y)) -
  U(z-(x-y)) \right]\\
=  -f(x+y) - 2 U(z) + U(z+y)  + U(z-x)  \\
+ \frac{1}{2} \left[U(z+ (x-y)) - U(z-(x-y)) \right] \\
 =  -f(x+y) + f(x) + f(y) 
-\frac{1}{2}\left[U(z+x)+U(z-x) + U(z+y) + U(z-y)\right]  \\
+U(z+y)  + U(z-x)  
+ \frac{1}{2} \left[U(z+ (x-y)) - U(z-(x-y)) \right] \\
 =  -f(x+y) + f(x) + f(y) 
+\frac{1}{2}\left[U(z+y) + U(z-x)  - U(z+x) - U(z-y)\right] \\
 + \frac{1}{2} \left[U(z+ (x-y)) - U(z-(x-y)) \right]\\
\end{split}\]

Let $y'=-y$. Then by Eventual Linearity we can set $z$ such that 
\[
U(z-y') + U(z-x)  - U(z+x) - U(z+y') + U(z+ (x+y')) - U(z-(x+y'))
=0.\]
Thus \[
 f(x-y) - f(x) - f(y)   = A = -f(x+y) + f(x) + f(y) .\]
\end{proof}

The function $f$ is continuous, and uniquely identified from $U$. Then Lemma~\ref{lem:one} and Proposition 4 of Chapter 11 of \cite{aczel1989functional} implies that there is a unique function $S:\Re^{2n}\rightarrow \Re$ such that $S$ is symmetric, bi-linear, and $f(x)=S(x,x)$.

\begin{lemma}\label{lem:two}
\[
g(x+y) = g(x) + g(y)
\]  
\end{lemma}
\begin{proof} First note that  $g(x+y) -g(x) -g(y)  = U(x+y) - U(x) - U(y) -( f(x+y) - f(x) - f(y) )$, and that Lemma~\ref{lem:one} implies that for any choice of $z,z',z''$:
\begin{align*}
 - (f(x+y) - f(x) - f(y)) & =  f(x-y) - f(x) - f(y) \\
& = \frac{1}{2}\left[U(z+x-y) - U(z) \right] + \frac{1}{2}\left[U(z - (x-y) ) - U(z) \right] \\
& - \left( \frac{1}{2}\left[U(z'+x) - U(z') \right] + \frac{1}{2}\left[U(z' - x  ) - U(z') \right] \right) \\
& - \left( \frac{1}{2}\left[U(z''+y) - U(z'') \right] + \frac{1}{2}\left[U(z'' - y  ) - U(z'') \right] \right).
\end{align*}

Therefore:
\begin{align*}
g(x+y) - g(x) - g(y) & =
U(x+y) - U(x) - U(y)   + \frac{1}{2}\left[U(z+x-y) - U(z) \right] \\
& + \frac{1}{2}\left[U(z - x
                       + y) - U(z) \right] \\
& - \frac{1}{2}\left[U(z'+x) - U(z') \right] - \frac{1}{2}\left[U(z' -
  x) - U(z') \right] \\
& - \frac{1}{2}\left[U(z''+y) - U(z'') \right] - \frac{1}{2}\left[U(z''  - y) - U(z'') \right].
\end{align*}

In particular, for $z'=y$ and $z''=x$, and using that $U(0)=0$, we
obtain that 
\begin{align*}
g(x+y) - g(x) - g(y) = & 
\frac{1}{2}\left[U(z+x-y) - U(z) \right] + \frac{1}{2}\left[U(z - x +
  y) - U(z) \right] \\
& - \frac{1}{2}\left[U(y - x)  - U(0) \right] - \frac{1}{2}\left[U(x - y) - U(0) \right] \\
= & 0, \end{align*} by status quo independence.
\end{proof}

Note that $g$ is continuous because $U$ is continuous. Then~\ref{lem:two} implies that $g$ is a linear function by Corollary 2 of Chapter 4 of \cite{aczel1989functional}.

For necessity: Status-quo independence is a simple calculation. Eventual Linearity is established by the following calculation. 
\begin{align*}
U(w+(x+y))-U(w-(x+y)) & = g(w+(x+y))-g(w-(x+y))  \\  & + S(w+(x+y),w+(x+y)) \\
& - S(w-(x+y),w-(x+y)) \\
& = 2g(x) + 2g(y) + 2S(w,x+y) + S(x+y,x+y) \\ & + 2S(w,x+y) - S(x+y,x+y) \\
& =2g(x) + 2g(y) + 4S(w,x) + 4S(w,y) \\ %+ S(x,x) \\ & + 2S(x,y) + S(y,y) \\
& = [g(w) + g(x) + S(w,w) + 2S(w,x) +S(x,x)] \\ & -[g(w)+g(-x) + S(w,w) + 2S(w,-x) + S(-x,-x)] \\
& + [g(y) + g(w) + S(w,w) + 2S(w,y) + S(y,y)] \\ &- [g(-y) + g(w) + S(w,w) + 2S(w,-y) + S(-y,-y)] \\
& = U(w+x) - U(w-x) + U(w+y) - U(w-y)
\end{align*}

%Observe that $x+z =a+c$ and $y+w=b+c$ for some $c$.  So $(x+z)\cdot (x+z) = (a+c)\cdot (a+c)$ and $(y+w)\cdot (y+w)=(b+c)\cdot (b+c)$.  Hence $(x+z)\cdot (x+z) - (y+w)\cdot (y+w) = (a+c)\cdot (a+c) -(b+c)\cdot (b+c)$.  

\section{Proof of Proposition~\ref{prop:homoth}}\label{sec:proofprophomoth}

We first establish the result for $w = 0$, so suppose that $\|x\|=\|y\|$, where $x\succeq y$.

We first establish the result for positive integer $\beta$. The proof proceeds by induction. Let $a\in\Re^n$ for which $a \perp x$ and $a \perp y$, further $\|a\|=\|x\|=\|y\|$. Such $a$ exists because $n\geq 3$.

By SOIO, it follows that $2a + x \succeq 2a + y$.  Further, $(x-a)\perp (x+a)$ and $(y-a)\perp (y+a)$.  Since $a + (x-a) \succeq a + (y-a)$ and $a + (x+a) \succeq a + (y+a)$, SOIO implies that $a + (x-a)+(x+a) \succeq a + (y-a)+(y+a)$, or $a +2x \succeq a + 2y$.  By SOIO, if $2y \succ 2x$, we would have $a + 2y \succ a + 2x$, a contradiction.  So, in fact $2x \succeq 2y$.

Suppose now that $x\succeq y$, and that we have shown $kx \succeq ky$ for $k\in\Na$.  We claim that $(k+1) x \succeq (k+1)y$.  By $(k+1)a\perp kx$, $(k+1)a\perp ky$ and SOIO, $(k+1)a + kx \succeq (k+1)a + ky$ (or $a + (kx +ka) \succeq a + (ky + ka)$). Moreover, $a + (x-a) \succeq a + (y-a)$.  Observe that $(kx+ka) \perp (x-a)$ and $(ky+ka) \perp (y-a)$.  Consequently, by SOIO, $a + (x-a) + k(x+a)\succeq a + (y-a) + k(y+a)$, or $ka + (k+1) x \succeq ka + (k+1)y$.  Again it must follow that $(k+1) x \succeq (k+1)y$.

By induction, $kx\succeq ky$ for all $k\in\Na$ with $k>0$. Note that the same argument shows that if $x\succ y$ then $kx \succ ky$. 

Now let $q>0$ be a rational number, $q=k/l$ with $k,l\in\Na$. Then it must hold that $qx\succeq qy$, as $qy\succ qx$ would imply that $lqy = ky \succ kx = lqx$ by the first step and the fact that $\norm{qx}=\norm{qy}$. 

Finally, by continuity of $\succeq$ we obtain that $\beta x\succeq \beta y$ for all real $\beta>0$. This proves the result for $w=0$. 

To see that the result holds for arbitrary $w$, it is enough to observe that the ranking $x \succeq_w y$ iff $(x+w) \succeq (y+w)$ satisfies SOIO and apply the previous argument.  

\bibliographystyle{ecta}
\bibliography{orthogonal}

\end{document}